\pdfoutput=1
\RequirePackage{ifpdf}
\ifpdf 
\documentclass[pdftex]{sigma}
\else
\documentclass{sigma}
\fi

\usepackage[all]{xy}

\numberwithin{equation}{section}

\newtheorem{Theorem}{Theorem}[section]
\newtheorem{Proposition}[Theorem]{Proposition}
 { \theoremstyle{definition}
\newtheorem{Remark}[Theorem]{Remark} }

\def\C{{\mathbb C}}
\def\P{{\mathbb P}}
\def\Z{{\mathbb Z}}

\def\o#1{{\overline{#1}}}
\def\u#1{{\underline{#1}}}

\begin{document}

\allowdisplaybreaks

\newcommand{\arXivNumber}{1706.10087}

\renewcommand{\PaperNumber}{092}

\FirstPageHeading

\ShortArticleName{A Variation of the $q$-Painlev\'e System with Af\/f\/ine Weyl Group Symmetry of Type $E_7^{(1)}$}

\ArticleName{A Variation of the $\boldsymbol{q}$-Painlev\'e System\\ with Af\/f\/ine Weyl Group Symmetry of Type $\boldsymbol{E_7^{(1)}}$}

\Author{Hidehito NAGAO}

\AuthorNameForHeading{H.~Nagao}

\Address{Department of Arts and Science, National Institute of Technology, Akashi College,\\ Hyogo 674-8501, Japan}
\Email{\href{mailto:nagao@akashi.ac.jp}{nagao@akashi.ac.jp}}

\ArticleDates{Received July 03, 2017, in f\/inal form November 24, 2017; Published online December 10, 2017}

\Abstract{Recently a certain $q$-Painlev\'e type system has been obtained from a reduction of the $q$-Garnier system. In this paper it is shown that the $q$-Painlev\'e type system is associated with another realization of the af\/f\/ine Weyl group symmetry of type $E_7^{(1)}$ and is dif\/ferent from the well-known $q$-Painlev\'e system of type $E_7^{(1)}$ from the point of view of evolution directions. We also study a connection between the $q$-Painlev\'e type system and the $q$-Painlev\'e system of type $E_7^{(1)}$. Furthermore determinant formulas of particular solutions for the $q$-Painlev\'e type system are constructed in terms of the terminating $q$-hypergeometric function.}


\Classification{14H70; 33D15; 33D70; 34M55; 37K20; 39A13; 41A21}

\section{Introduction}\label{sect:intro}

\subsection{Background}\label{subsec:Back}
In \cite{Sakai01} H.~Sakai has classif\/ied the second order continuous and discrete Painlev\'e equations into 22 cases by using the geometric theory of certain rational surfaces, called the``{\it spaces of initial values}''\footnote{For each of the six continuous Painlev\'e equations, K.~Okamoto has constructed certain rational surfaces, called the ``{\it spaces of initial values}'', which parametrize all the solutions \cite{Okamoto79}.}, connected to af\/f\/ine root systems. The spaces of initial values are obtained from $\P^1 \times \P^1$ (resp.~$\P^2$) by blowing up at~8 (resp.~9) singular points. In view of the conf\/iguration of~8 (resp.~9) singular points in~$\P^1 \times \P^1$ (resp.~$\P^2$), there exist three types of discrete Painlev\'e equations and six continuous Painlev\'e equations in the classif\/ication: elliptic dif\/ference \smash{($e$-)}, multiplicative dif\/ference ($q$-), additive dif\/ference ($d$-) and continuous (dif\/ferential). Each of these Painlev\'e equations is constructed in a unif\/ied manner as the bi-rational action of a~translation part of the corresponding af\/f\/ine Weyl group symmetry on a certain family of the rational surfaces. The sole $e$-Painlev\'e equation \cite{ORG01} having the af\/f\/ine Weyl group symmetry of type $E_8^{(1)}$ is obtained from the most generic conf\/iguration on the unique curve of bi-degree $(2,2)$ called the smooth ``{\it elliptic curve}''. All of the other Painlev\'e equations are derived from its degeneration. For instance, the $q$-Painlev\'e system with the symmetry of type $E_7^{(1)}$ is well known to be obtained from a~conf\/iguration of eight singular points on two curves of bi-degree $(1,1)$ in $\P^1 \times \P^1$. The second order continuous and discrete Painlev\'e equations are classif\/ied into the 22 cases\footnote{Some $q$-Painlev\'e equations, such as a second order case of the system \cite{KNY02-2} (see also~\cite{Takenawa03}), does not belong to the list of discrete Painlev\'e equations appeared in~\cite{Sakai01}.} according to the degeneration diagram of af\/f\/ine Weyl group symmetries (see Fig.~\ref{fig1}),
\begin{figure}[t]\small\centering
$\quad
\xymatrix @C=7pt@R=4pt@M=2pt{
{\rm ell.\mbox{($e$-)}} & E_8^{(1)}\ar[d] & & & & & & & \mbox{$\substack{A_1^{(1)}\\|\alpha|^2=8}$}\ar[rd] \\
{\rm mul.}\mbox{($q$-)} & E_8^{(1)}\ar[r]\ar[d] & E_7^{(1)}\ar[r]\ar[d] & E_6^{(1)}\ar[r]\ar[d] & D_5^{(1)}\ar[r]\ar[rd] & A_4^{(1)}\ar[r]\ar[rd] & (A_2+A_1)^{(1)}\ar[r]\ar[rd]\ar[rdd] & (A_1+\substack{A_1\\|\alpha|^2=14})^{(1)}\ar[r]\ar[rd]\ar[ru] & A_1^{(1)}\ar[rd]\ar@/^23pt/[dd] & A_0^{(1)}\ar@/^23pt/[dd] \\
{\rm add.}\mbox{($d$-)}& E_8^{(1)}\ar[r] & E_7^{(1)}\ar[r] & E_6^{(1)}\ar[rr] & & \underset{(\mbox{\footnotesize$P_{\rm VI}$})}{D_4^{(1)}}\ar[r] &\underset{(\mbox{\footnotesize$P_{\rm V}$})}{A_3^{(1)}}\ar[r]
\ar[rd] &\underset{(\mbox{\footnotesize $P_{\rm III}$})}{2(A_1)^{(1)}}\ar[r]\ar[rd] & \underset{(\mbox{\footnotesize $P_{\rm III}^{D_7^{(1)}}$})}{\substack{A_1^{(1)}\\|\alpha|^2=4}}\ar[r] & \underset{(\mbox{\footnotesize$P_{\rm III}^{D_8^{(1)}}$})}{A_0^{(1)}} &&\\
&&&&&&& \underset{(\mbox{\footnotesize$P_{\rm IV}$})}{A_2^{(1)}}\ar[r] & \underset{(\mbox{\footnotesize$P_{\rm II}$})}{A_1^{(1)}}\ar[r] & \underset{(\mbox{\footnotesize$P_{\rm I}$})}{A_0^{(1)}}&&&
}\!\!\!\!
$
\caption{}\label{fig1}
\end{figure}
where the symbol $A \to B$ represents that $B$ is obtained from $A$ by a certain limiting procedure. The $d$-Painlev\'e equation of type $D_4^{(1)}$ and its degeneration (expect for $A_0^{(1)}$) arise as B\"acklund (Schlesinger) transformations of the six continuous Painlev\'e equations\footnote{$P_{\rm III}^{D_i^{(1)}}$ symbolizes $P_{\rm III}$ having the surface connected to the af\/f\/ine root system of type $D_i^{(1)}$.} ($P_{\rm I}$, $\ldots$,$P_{\rm VI}$). The symbol $\substack{A_1^{(1)}\\|\alpha|^2=l}$ means the root subsystem of type $A_1^{(1)}$ whose square length of roots is $l$.

Similarly to the dif\/ferential Painlev\'e systems, the discrete Painlev\'e systems are known to have particular solutions expressed by various hypergeometric functions \cite{KMNOY03, KMNOY04, KMNOY05, KNY17, Masuda09, Masuda11,RGTT01}. The particular solutions of the elliptic Painlev\'e equation are expressed in \cite{KMNOY03} in terms of the elliptic hypergeometric function ${}_{10}E_9$ \cite{GaR04}. In the case of $q$-$E_7^{(1)}$, the particular solutions are expressed in \cite{Nagao15} in terms of the terminating $q$-hypergeometric function ${}_4\varphi_3$,\footnote{The terminating balanced ${}_4\varphi_3$ is rewritten into the terminating $q$-hypergeometric function ${}_8W_7$ by Watson's transformation formula \cite{GaR04}. For particular solutions in terms of ${}_8W_7$, see \cite{KMNOY04, KMNOY05, Masuda09}.} where the function~${}_k\varphi_l$~\cite{GaR04} is def\/ined by
\begin{gather}\label{eq:HGF}
{}_k\varphi_l\left(
\begin{matrix}
\alpha_1,&\dots,&\alpha_k\\
\beta_1,&\dots,&\beta_l
\end{matrix}
,x
\right)
=\sum_{s=0}^{\infty}\frac{(\alpha_1,\ldots,\alpha_k)_s}{(\beta_1,\ldots,\beta_l,q)_s}\big[(-1)^sq^{\left(\substack{s\\2}\right)}\big]^{1+l-k}x^s,
\end{gather}
with $\left(\substack{s\\2}\right)=\frac{s(s-1)}{2}$. Here the standard $q$-Pochhammer symbol\footnote{Actually Pochhammer himself used the symbol $(a)_n$ not as a rising shifted factorial but as a binomial coef\/f\/icient \cite{Knuth92}.} is def\/ined by
\begin{gather*}
(x)_\infty:=\prod_{i=0}^\infty (1-q^i x), \qquad
(x)_s:=\frac{(x)_\infty}{(xq^s)_\infty}, \qquad
(x_1,x_2, \ldots, x_k)_s:=(x_1)_s(x_2)_s \cdots (x_k)_s.
\end{gather*}

It is common to nonlinear integrable systems that they arise as the compatibility condition of linear equations and their deformed equations. The pair of the linear equations is called a ``{\it Lax pair}'' for the nonlinear system. Similarly to the continuous Painlev\'e equations \cite{JM81-1,JM81-2,JM81-3, OKSO06, OO06}, Lax pairs for the discrete Painlev\'e equations have been studied from various points of view in \cite{GORS98, JS96, KNY17, Murata09, Rains11,Sakai06,Yamada09-2, Yamada11}. For instance, as a geometric approach, the Lax pair for the $e$-Painlev\'e equation has been formulated in \cite{Yamada09-2} as a curve of bi-degree $(3,2)$ in $\P^1\times\P^1$ passing through 12 points. In the case of $q$-$E_7^{(1)}$, the Lax pair has been similarly formulated in \cite{KNY17, Yamada09-2}.

In \cite{Sakai05-1} the $q$-Garnier system was formulated as a multivariable extension of the well-known $q$-$P_{\rm VI}$ (i.e., $q$-$D_5^{(1)}$) system \cite{JS96} by H.~Sakai, and has recently been studied in \cite{NY16, OR16-1, Sakai05-2}\footnote{For the related works, see \cite{DST13, DT14, OR16-1} (additive Garnier system), \cite{OR16-2, Yamada17} (elliptic Garnier system).}. In~\cite{NY16} a Lax pair, an evolution equation and two kinds of particular solutions\footnote{These solutions have been constructed in terms of the $q$-Appell Lauricella function (resp.\ the generalized $q$-hypergeometric function) in \cite{NY16, Sakai05-2} (resp.~\cite{NY16}).
} for the $q$-Garnier system have been simply expressed by applying a certain method of Pad\'e approximation and its analogue (i.e., Pad\'e interpolation)\footnote{The Pad\'e method has been also applied to the continuous/discrete Painlev\'e systems in \cite{Ikawa13, Nagao15,Nagao17-1, Nagao16, Nagao17-2, NTY13, Yamada09-1, Yamada14}. For the case of $q$-$E_7^{(1)}$ \cite{Nagao15}, see Section~\ref{sec:Pade}. For the works related to the dif\/ferential Garnier system, see \cite{Mano12, Yamada09-1} (Pad\'e approximation), \cite{MT14, MT17} (Hermite--Pad\'e approximation).}. The $q$-$D_5^{(1)}$ (resp.~$q$-$E_6^{(1)}$) system appears as a reduction of case $N=1$ \cite{Sakai05-1} (resp.\ particular case of $N=2$ \cite{Sakai06}) of the $q$-Garnier system having $2N$ dependent variables. Recently the $q$-Painlev\'e type system\footnote{As another derivation of the equations (\ref{eq:E7T1_ev}) and (\ref{eq:E7T1_L}), see Appendix~\ref{subsec:Pade_Lax} (Pad\'e interpolation method).} \cite[Section~2.5]{NY16} has appeared as a~particular case of $N=3$. However the $q$-$E_8^{(1)}$ system has been not obtained from a reduction of the $q$-Garnier system.
\begin{Remark}\label{rem:variation}
We call a certain $q$-Painlev\'e type system ``{\it a variation of a $q$-Painlev\'e system}'' having a well-known direction\footnote{In case of $q$-$E_7^{(1)}$, $T_2$ (\ref{eq:E7T2_shift}) is the well-known direction and $T_1$ (\ref{eq:E7T1_shift}) is a variation direction.}, when both systems satisfy the following: (i)~They are associated with dif\/ferent realizations of the symmetry/surface of the same type in the Sakai's classif\/ication. (ii)~Their time evolutions are dif\/ferent from the viewpoint of shift operator on parameters.
\end{Remark}

\subsection{Purpose and organization}

Our main subject is the $q$-Painlev\'e type system \cite{NY16} regarded as a variation of the well-known $q$-$E_7^{(1)}$ system \cite{GR99}. The purpose of this paper is the following three.
\begin{itemize}\itemsep=0pt
\item We show that the $q$-Painlev\'e type system is a bi-rational transformation and is related to a~novel realization of the symmetry/surface of type $E_7^{(1)}$/$A_1^{(1)}$. Then it is clarif\/ied to be a~variation of the $q$-$E_7^{(1)}$ system.
\item The Lax pair for the $q$-Painlev\'e type system is obtained from a certain reduction of the $q$-Garnier system and we study a connection between the $q$-Painlev\'e type system and the~$q$-$E_7^{(1)}$ system by comparing their Lax equations.

\item Particular solutions for the $q$-Painlev\'e type system are given as a reduction of the $q$-Garnier system.
\end{itemize}

This paper is organized as follows. In Section~\ref{sec:E7T1} we prove that the $q$-Painlev\'e type system is a bi-rational transformation and we investigate its conf\/iguration of 8 singular points on a~curve of bi-degree $(2,2)$ in coordinates $(f,g)$ $\in\P^1 \times \P^1$. In Section~\ref{sec:Lax} we brief\/ly recall Lax equations for the $q$-Garnier system \cite[Section~2.1]{NY16}, and study a reduction of particular case $N = 3$ of the $q$-Garnier system. Consequently, we obtain Lax equations for the $q$-Painlev\'e type system. In Section~\ref{sec:L1} the Lax equation for the $q$-Painlev\'e type system is uniquely determined by a~characterization, and recall the characterization of Yamada's Lax equation for the $q$-$E_7^{(1)}$ system. Then we investigate a connection between these systems by comparing characterizations of their Lax equations. In Section~\ref{sec:Ps} we recall the particular solutions of the $q$-Garnier system and we construct particular solutions of the $q$-Painlev\'e type system by applying a reduction. In Appendix \ref{sec:Pade} we derive the $q$-Painlev\'e type system, its Lax pair and its particular solutions by using a Pad\'e interpolation.

\section[$q$-Painlev\'e type system]{$\boldsymbol{q}$-Painlev\'e type system}\label{sec:E7T1}
In this section we f\/irst recall the $q$-Painlev\'e type system \cite[Section~2.5]{NY16}. Then we prove that the system is a bi-rational transformation and conf\/irm that the system has the symmetry/surface of type $E_7^{(1)}$/$A_1^{(1)}$ by its conf\/iguration of eight singular points. Let $q$ ($|q|<1$), $a_1,\ldots,a_4$, $b_1,\ldots,b_4$, $c_1$ and~$d_1$ $\in\C^{\times}$ be complex parameters with a constraint $\prod\limits_{i=1}^{4}\frac{a_i}{b_i}=q\frac{c_1^2}{d_1^2}$, and let $(f, g) \in\P^1 \times \P^1$ be dependent variables. Def\/ine $T_a\colon a \to qa$ by a $q$-shift operator of parameter $a$. Then we consider a $q$-shift operator $T_1$\footnote{The operator $T_1$ is generally selected as $T_{a_i}^{-1}T_{b_j}^{-1}$. The directions such as $T_{a_i}^{-1}T_{b_j}^{-1}$, $T_{a_i}^{-1}T_{a_j}^{-1}T_{c_k}^{-1}$ and $T_{a_i}^{-1}T_{a_j}^{-1}T_{d_k}$ are fundamental ones.}
\begin{gather}\label{eq:E7T1_shift}
T_1=T_{a_1}^{-1}T_{b_1}^{-1}.
\end{gather}
Here for any object $X$ the corresponding shifts are denoted as $\o{X}:=T_1(X)$ and $\u{X}:=T_1^{-1}(X)$. The operator $T_1$ plays the role of the evolution of the system. The system is described by the following transformation $T_1^{-1}(g)=\u{g}(f,g)$ and $T_1(f)=\o{f}(f,g)$ in $\mathbb{P}^1 \times\mathbb{P}^1$:
\begin{gather}
\big(e_1f^2+e_1fg+c_1\big)\left(\frac{e_1}{q}f^2+\frac{e_1}{q}\u{g}f+c_1\right)=c_1^2\frac{\prod\limits_{i=2}^{4}(1-a_if)(1-b_if)}{(1-a_1f)(1-b_1f)},\nonumber\\
\frac{x_1^2(1-fx_1)(1-\o{f}x_1)}{x_2^2(1-fx_2)(1-\o{f}x_2)}=\prod_{i=2}^{4}\frac{(x_1-a_i)(x_1-b_i)}{(x_2-a_i)(x_2-b_i)}.\label{eq:E7T1_ev}
\end{gather}
Here $e_1=d_1a_2a_3a_4b_1^{-1}$, and $x=x_1, x_2\big(=\frac{e_1}{c_1x_1}\big)$ are solutions of an equation $\varphi=0$ where
\begin{gather}\label{eq:E7T1_phi}
\varphi(x)=e_1+e_1gx+c_1x^2.
\end{gather}
Then we have
\begin{Proposition}\label{prop:BC}
The $q$-Painlev\'e type system \eqref{eq:E7T1_ev} has the following properties:
\begin{enumerate}\itemsep=0pt
\item[$(i)$] It is a bi-rational transformation $T_1^{-1}(g)=\u{g}(f,g)$ and $T_1(f)=\o{f}(f,g)$ $\in \mathbb{P}^1 \times\mathbb{P}^1$.
\item[$(ii)$] It is associated with a novel realization \eqref{eq:E7T1_8p} of the symmetry/surface of type $E_7^{(1)}$/$A_1^{(1)}$.
\end{enumerate}
\end{Proposition}

\begin{proof}(i) It is easy to see that the f\/irst equation of (\ref{eq:E7T1_ev}) is a rational transformation $T_1^{-1}(g)=\u{g}(f,g)$. The second equation of (\ref{eq:E7T1_ev}) is rewritten as
\begin{gather}\label{eq:E7T1_fu}
\o{f}=\frac{x_1^2(1-fx_1){\mathcal A}(x_2) -x_2^2(1-fx_2){\mathcal A}(x_1)}{x_1^3(1-fx_1){\mathcal A}(x_2) -x_2^3(1-fx_2){\mathcal A}(x_1)},
\end{gather}
where ${\mathcal A}(x)=\prod\limits_{i=2}^4 (x-a_i)(x-b_i)$.
Then the numerator and denominator of (\ref{eq:E7T1_fu}) are alternating with respect to $x_1 \leftrightarrow x_2=\frac{e_1}{c_1 x_1}$, and Laurent polynomials $\sum\limits_{i=-4}^{4}h_i x_1^i$ where $h_i$ is depending on $f$, $a_i$, $b_j$, $c_1$ and $d_1$. Accordingly the numerator and denominator are expressed by Laurent polynomials $(x_1-x_2)\sum\limits_{i=0}^{3}\tilde{h}_i(x_1+x_2)^i$, where~$\tilde{h}_i$ is depending on $f$, $a_i$, $b_j$, $c_1$ and $d_1$. Due to the relation $x_1+x_2=-\frac{e_1g}{c_1}$, the transformation $T_1(f)=\o{f}(f,g)$ is given as a rational polynomial of bi-degree $(1,3)$ in $(f,g)$. Therefore the property (i) is proved. (ii) Eight singular points $(f_s, g_s) \in \mathbb{P}^1 \times \mathbb{P}^1$ $(s=1,\ldots,8)$ in coordinates $(f,g)$ are on one line $g=\infty$ of bi-degree $(0,1)$ and one parabolic curve $e_1f^2+e_1fg+c_1=0$ of bi-degree $(2,1)$ as follows:
\begin{gather}\label{eq:E7T1_8p}
\left(\frac{1}{a_1}, \infty\right),\quad \left(\frac{1}{b_1}, \infty\right), \quad \left(\frac{1}{a_i}, -\frac{1}{a_i}-\frac{a_ic_1}{e_1}\right)_{i=2,3,4}, \quad \left(\frac{1}{b_i}, -\frac{1}{b_i}-\frac{b_ic_1}{e_1}\right)_{i=2,3,4}.
\end{gather}
Hence the property (ii) is conf\/irmed since the conf\/iguration (\ref{eq:E7T1_8p}) is the realization of the surface type~$A_1^{(1)}$.
\end{proof}

According to Remark \ref{rem:variation}, the system (\ref{eq:E7T1_ev}) is regarded as a variation of the $q$-$E_7^{(1)}$ system.

\section{Lax equations}\label{sec:Lax}

In this section we recall Lax equations for the $q$-Garnier system \cite[Section~2.1]{NY16} and investigate a reduction of particular case $N=3$ of them. As a result, we obtain the Lax equations for the $q$-Painlev\'e type system (\ref{eq:E7T1_ev}).

\subsection[Case of the $q$-Garnier system]{Case of the $\boldsymbol{q}$-Garnier system}\label{subsec:Lax_Gar}

The scalar Lax equations for the $q$-Garnier system are
\begin{gather}
 L_1(x)=A(x)F\left(\frac{x}{q}\right)y(q x)+qc_1c_2B\left(\frac{x}{q}\right)F(x)y\left(\frac{x}{q}\right)\nonumber\\
\hphantom{L_1(x)=}{} -\left\{(x-a_1)(x-b_1)F\left(\frac{x}{q}\right)G(x)+\frac{F(x)}{G\big(\frac{x}{q}\big)}V\left(\frac{x}{q}\right)\right\}y(x),\nonumber\\
 L_2(x)=F(x)\o{y}(x)-A_1(x)y(qx)+(x-b_1)G(x)y(x),\nonumber\\
L_3(x)=\o{F}\left(\frac{x}{q}\right)y(x)+(x-a_1)G\left(\frac{x}{q}\right)\o{y}(x)-qc_1c_2B_1\left(\frac{x}{q}\right)\o{y}\left(\frac{x}{q}\right),\label{eq:Gar_L}
\end{gather}
where
\begin{gather}
A(x)=\prod_{i=1}^{N+1}(x-a_i), \qquad B(x)=\prod_{i=1}^{N+1}(x-b_i), \qquad A_1(x)=\frac{A(x)}{x-a_1}, \qquad B_1(x)=\frac{B(x)}{x-b_1},\nonumber\\
F(x)=\sum_{i=0}^N f_i x^i, \qquad G(x)=\sum_{i=0}^{N-1} g_i z^{i}, \qquad V(x)=qc_1c_2A_1(x)B_1(x)-F(x)\o{F}(x).\label{eq:Gar_ABFG}
\end{gather}
Here the deformation direction is $T_1$ (\ref{eq:E7T1_shift}) and $f_0, \ldots, f_N$, $g_0, \ldots, g_{N-1}$ $\in \P^{1}$ are variables depending on parameters $a_i$, $b_i$, $c_i$, $d_i$ with a constraint $\prod\limits_{i=1}^{N+1}\frac{a_i}{b_i}=q\frac{c_1c_2}{d_1d_2}$.
\begin{Remark}\label{rem:L1L2L3}
The scalar Lax pair $L_1=0$ and $L_2=0$ (or $L_3=0$) is equivalent to the pair of the deformation equations $L_2=0$ and $ L_3=0$.
\end{Remark}

The equation $L_1=0$ (we call it the $L_1$ equation) is equivalent to one for Sakai's system given in~\cite{Sakai05-1} and the deformation direction is opposite to one for Sakai's system. The $q$-Garnier system is
\begin{gather}
G(x)\u{G}(x)=c_1c_2\frac{A_1(x)B_1(x)}{(x-a_1)(x-b_1)} \qquad {\rm for} \quad F(x)=0,\nonumber\\
F(x)\o{F}(x)=qc_1c_2A_1(x)B_1(x) \qquad {\rm for} \quad G(x)=0,\label{eq:Gar_ev}\\
f_N\o{f}_N=q (g_{N-1}-c_1)(g_{N-1}-c_2),\qquad f_0\o{f}_0=a_1b_1\left(g_0-\frac{d_1}{a_1b_1}A(0)\right)\left(g_0-\frac{d_2}{a_1b_1}A(0)\right),\nonumber
\end{gather}
where $2N$ variables $\frac{f_1}{f_0},\ldots,\frac{f_N}{f_0}$, $g_0,\ldots,g_{N-1}$ are the dependent variables. Then we have the following fact\footnote{For the proof of Proposition~\ref{prop:Gar_com}, see \cite[Section~2.3]{NY16}.}.
\begin{Proposition}\label{prop:Gar_com}
The compatibility condition of the Lax pair $L_1=0$ and $L_2=0$ \eqref{eq:Gar_L} is equivalent to the $q$-Garnier system \eqref{eq:Gar_ev}.
\end{Proposition}

\subsection[Reduction to the $q$-Painlev\'e type system]{Reduction to the $\boldsymbol{q}$-Painlev\'e type system}\label{subsec:Lax_E7T1}

We impose a reduction condition by a constraint of the parameters
\begin{gather}\label{eq:E7T1_cons}
c_1=c_2, \qquad d_1=d_2
\end{gather}
and specialize the dependent variables as
\begin{gather}\label{eq:E7T1_rn}
f_0=f_3=0, \qquad f_1=w_1, \qquad f_2=-fw_1, \qquad g_0=e_1, \qquad g_1=e_1g, \qquad g_2=c_1,
\end{gather}
where $w_1$\footnote{For convenience, $f_1$ is replaced by a dif\/ferent symbol $w_1$ since $f_1$ looks like $f$.} is a ``{\it gauge freedom}''. Applying the conditions (\ref{eq:E7T1_cons}) and (\ref{eq:E7T1_rn}) into (\ref{eq:Gar_L}) and (\ref{eq:Gar_ev}), we obtain the following linear equations
\begin{gather}
L_1(x)=A(x)\left(1-\frac{fx}{q}\right)y(qx)+q^2c_1^2B\left(\frac{x}{q}\right)(1-fx)y\left(\frac{x}{q}\right)\nonumber\\
\phantom{L_1(x)=}{} -\left\{(x-a_1)(x-b_1)\left(1-\frac{fx}{q}\right)\varphi(x)+\frac{q(1-fx)}{\varphi\big(\frac{x}{q}\big)}V\left(\frac{x}{q}\right)\right\}y(x), \nonumber\\
L_2(x)=w_1x(1-fx)\o{y}(x)-A_1(x)y(qx)+(x-b_1)\varphi(x)y(x),\nonumber\\
L_3(x)=\o{w}_1\frac{x}{q}\left(1-\o{f}\frac{x}{q}\right)y(x)+(x-a_1)\varphi\left(\frac{x}{q}\right)\o{y}(x)-qc_1^2B_1\left(\frac{x}{q}\right) \o{y}\left(\frac{x}{q}\right),\label{eq:E7T1_L}
\end{gather}
where $\varphi$ is given by (\ref{eq:E7T1_phi}) and
\begin{gather}
A(x)=\prod_{i=1}^{4}(x-a_i), \qquad B(x)=\prod_{i=1}^{4}(x-b_i), \nonumber\\
V(x)=qc_1^2A_1(x)B_1(x)-w_1\o{w}_1x^2(1-fx)(1-\o{f}x).\label{eq:E7T1_ABV}
\end{gather}
Then we have
\begin{Proposition}\label{prop:E7T1_com}
 The compatibility condition of the $L_1$ and $L_2$ equations \eqref{eq:E7T1_L} is equivalent to the $q$-Painlev\'e type system \eqref{eq:E7T1_ev}.
\end{Proposition}
\begin{proof}
Thanks to Proposition \ref{prop:Gar_com} and the conditions (\ref{eq:E7T1_cons}), (\ref{eq:E7T1_rn}).
\end{proof}

The pair of $L_1$ and $L_2$ equations \eqref{eq:E7T1_L} is regard as the Lax pair for the system (\ref{eq:E7T1_ev}).

\section[Characterization of the $L_1$ equation]{Characterization of the $\boldsymbol{L_1}$ equation}\label{sec:L1}

In \cite{Yamada11} Y.~Yamada has formulated a Lax form for the $q$-Painlev\'e equation of type $E_7^{(1)}$ as the linear equation (say $L_1=0$) and its deformed equation. Our direction $T_1$ (\ref{eq:E7T1_shift}) is dif\/ferent from Yamada's one. In general the $L_1$ equation is expressed in terms of dif\/ferent dependent variables according to several deformation directions. In this section, from the viewpoint of coordinates of dependent variables, we study a connection between our $L_1$ equation~(\ref{eq:E7T1_L}) and Yamada's $L_1$ equation.

\subsection[Case of our $L_1$ equation]{Case of our $\boldsymbol{L_1}$ equation}\label{subsec:L1_E7T1}
We consider characterizing our $L_1$ equation (\ref{eq:E7T1_L}) in the coordinates $(f,g) \in \mathbb{P}^1 \times \mathbb{P}^1$. The compatibility of $L_2$ and $L_3$ equations (\ref{eq:E7T1_L}) gives the f\/irst equation of (\ref{eq:E7T1_ev}) and two relations
\begin{gather}\label{eq:E7T1_w1}
w_1\o{w}_1x_1^2(1-fx_1)(1-\o{f}x_1)=qc_1^2A_1(x_1)B_1(x_1),\\
\label{eq:E7T1_w2} w_1\o{w}_1x_2^2(1-fx_2)(1-\o{f}x_2)=qc_1^2A_1(x_2)B_1(x_2).
\end{gather}
The relations (\ref{eq:E7T1_w1}) and (\ref{eq:E7T1_w2}) give the second equation of (\ref{eq:E7T1_ev}). Eliminating $\o{f}$, $w_1$ and $\o{w}_1$ from the expression $L_1$ (\ref{eq:E7T1_L}) by using the f\/irst equation of (\ref{eq:E7T1_ev}) and the relation (\ref{eq:E7T1_w1}), the expres\-sion~$L_1$~(\ref{eq:E7T1_L}) is rewritten in terms of variables $f$ and $x_1$ as the expression
\begin{gather}
L_1(x)=A(x)\left(1-f\frac{x}{q}\right)\left[y(qx)-\frac{(x-b_1)\varphi(x)}{A_1(x)}y(x)\right]\nonumber\\
\hphantom{L_1(x)=}{}+q^2c_1^2B\left(\frac{x}{q}\right)(1-fx)\left[y(\frac{x}{q})-\frac{A_1\big(\frac{x}{q}\big)}{\big(\frac{x}{q}-b_1\big) \varphi\big(\frac{x}{q}\big)}\right]\label{eq:E7T1_L1}\\
\hphantom{L_1(x)=}{} +\frac{c_1^2x^2(1-fx)\big(1-f\frac{x}{q}\big)}{(x_1-x_2)\varphi\big(\frac{x}{q}\big)}\!
\left[\frac{\big(\frac{x}{q}-x_2\big)A_1(x_1)B_1(x_1)}{x_1^2(1-fx_1)}-\frac{\big(\frac{x}{q}-x_1\big)A_1(x_2)B_1(x_2)}{x_2^2(1-fx_2)}\right]\!y(x),
\nonumber
\end{gather}
where $x_2=\frac{e_1}{c_1x_1}$. Next when we set the expression $L_1^*$ by
\begin{gather}\label{eq:E7T1_L1*}
L_1^*(x)=(1-fx_1)(1-fx_2)L_1(x).
\end{gather}
Then we have:
\begin{Proposition}\label{prop:E7T1_L1*}
The $L_1^*$ equation \eqref{eq:E7T1_L1*} has the following characterization\footnote{For other cases, see \cite{KNY17,Yamada09-2} (case $e$-$E_8^{(1)}$), \cite{KNY17,Yamada11} (case $q$-$E_8^{(1)}$).}:
\begin{enumerate}\itemsep=0pt
\item[$(i)$] The expression $L_1^*(f,g)$ is a polynomial of bi-degree $(3,2)$ in the coordinates $(f,g) \in \mathbb{P}^1 \times \mathbb{P}^1$.
\item[$(ii)$] As a polynomial, the expression $L_1^*(f,g)$ vanishes at the following $12$ points $(f_s, g_s) \in \mathbb{P}^1 \times \mathbb{P}^1$ $(s=1,\ldots,12)$:
\begin{gather} \left(\frac{1}{a_1}, \infty\right),\quad \left(\frac{1}{b_1}, \infty\right), \quad \left(\frac{1}{a_i}, -\frac{1}{a_i}-\frac{a_ic_1}{e_1}\right)_{i=2,3,4}, \quad \left(\frac{1}{b_i}, -\frac{1}{b_i}-\frac{b_ic_1}{e_1}\right)_{i=2,3,4},
\nonumber\\
 \left(\frac{q}{x}, \infty\right),\quad \left(\frac{1}{x}, -\frac{1}{x}-\frac{c_1x}{e_1}\right), \quad \left(\frac{1}{x}, g_{\frac{1}{x}}\right), \quad \left(\frac{q}{x}, g_{\frac{q}{x}}\right),\label{eq:E7T1_12p}
\end{gather}
where the first $8$ points are as in \eqref{eq:E7T1_8p} and $g_u$ is given by
\begin{gather*}
\frac{y\big(\frac{q}{u}\big)}{y\big(\frac{1}{u}\big)}=:\frac{\big(\frac{1}{u}-b_1\big)\big(e_1+\frac{e_1g_u}{u}+\frac{c_1}{u^2}\big)} {A_1\big(\frac{1}{u}\big)}, \qquad u=\frac{1}{x}, \frac{q}{x}.
\end{gather*}
\end{enumerate}
Conversely the $L_1^*$ equation is uniquely characterized by these properties $(i)$ and $(ii)$.
\end{Proposition}

\begin{proof}
By the expression $L_1$ (\ref{eq:E7T1_L1}), the expression $L_1^*$ is rewritten in terms of variables $f$ and~$x_1$ as follows
\begin{gather}
 L_1^*(x)=A(x)(1-fx_1)(1-fx_2)\left(1-f\frac{x}{q}\right)y(qx)\nonumber\\
 \hphantom{L_1^*(x)=}{}+q^2c_1^2(1-fx_1)(1-fx_2)(1-fx)y\left(\frac{x}{q}\right)+P(x)y(x),\label{eq:E7T1_LL1*}
\end{gather}
where
\begin{gather*}
 P(x)=\frac{q^2c_1^2(1-fx)\big(1-f\frac{x}{q}\big)\big(\frac{x}{q}\big)^2}{(x_1-x_2)\varphi\big(\frac{x}{q}\big)}Q(x)
\nonumber\\ \hphantom{P(x)=}{}-(x-a_1)(x-b_1)(1-fx_1)(1-fx_2)\left(1-f\frac{x}{q}\right)\varphi(x), \nonumber\\
Q(x)=\frac{(1-fx_2)\big(\frac{x}{q}-x_2\big)A_1(x_1)B_1(x_1)}{x_1^2}-\frac{(1-fx_1)\big(\frac{x}{q}-x_1\big)A_1(x_2)B_1(x_2)}{x_2^2}
\nonumber\\ \phantom{Q(x)=}{}-\frac{(x_1-x_2)(1-fx_1)(1-f x_2)A_1\big(\frac{x}{q}\big)B_1\big(\frac{x}{q}\big)}{\big(1-f\frac{x}{q}\big)\big(\frac{x}{q}\big)^2},
\end{gather*}
and $x_1$, $x_2\big({=}\frac{e_1}{c_1x_1}\big)$ are as in Section~\ref{subsec:Lax_E7T1}. Similarly to the proof of Proposition \ref{prop:BC}, the expression~$Q(x)$ is given as a Laurent polynomial $(x_1-x_2)\sum\limits_{i=0}^{3}k_i(x_1+x_2)^i$ where $k_i$ depends on $f$, $a_i$, $b_i$, $c_1$ and~$d_1$. The expression $Q(x)$ has zeros at $x=qx_1, qx_2$ which are solutions of the equation $\varphi\big(\frac{x}{q}\big)=0$. Therefore, due to the relation $x_1+x_2=-\frac{e_1g}{c_1}$, the expression $P(x)$ (i.e., the coef\/f\/icient of $y(x)$) is a polynomial of bi-degree $(3,2)$ in $(f,g)$. It is obvious that the coef\/f\/icients of $y(qx)$ and $y\big(\frac{x}{q}\big)$ are polynomials of bi-degree $(3,1)$ in $(f,g)$. Hence the property (i) is completely proved. Next the property (ii) can be easily conf\/irmed by substituting the 12~points~(\ref{eq:E7T1_12p}) into the expression~$L_1^*$~(\ref{eq:E7T1_LL1*}).
\end{proof}

\subsection[Case of Yamada's $L_1$ equation]{Case of Yamada's $\boldsymbol{L_1}$ equation}\label{subsec:L1_E7T2}
Firstly we recall the well-known $q$-$E_7^{(1)}$ system \cite{GR99} and the corresponding Yamada's Lax form~\cite{Yamada11}. Next we characterize Yamada's $L_1$ equation in terms of dependent variables. The complex parameters $a_i$, $b_i$, $c_1$ and $d_1$ be as in Section~\ref{sec:E7T1} and let $(\lambda, \mu) \in\P^1 \times \P^1$ be dependent variables. Then we consider a $q$-shift operator $T_2$ \footnote{The direction $T_2$ is given by a composition of the fundamental ones such as $T_1$ (\ref{eq:E7T1_shift}) and $T_{a_2}^{-1}T_{b_2}^{-1}$.} as
\begin{gather}\label{eq:E7T2_shift}
T_2=T_{a_1}^{-1}T_{a_2}^{-1}T_{b_1}^{-1}T_{b_2}^{-1}.
\end{gather}
The operator $T_2$ plays the role of the evolution of the $q$-$E_7^{(1)}$ system. The system is well-known as the bi-rational transformation $T_2^{-1}(\mu)=\u{\mu}(\lambda,\mu)$ and $T_2(\lambda)=\o{\lambda}(\lambda,\mu)$ in $\mathbb{P}^1 \times \mathbb{P}^1$ as follows \cite{GR99, KMNOY04, KNY17, Sakai01}
\begin{gather}
\frac{\big(\lambda\mu-\frac{e_2}{c_1}\big)\big(\lambda\u{\mu}-\frac{e_2}{qc_1}\big)}{(\lambda\mu-1)(\lambda\u{\mu}-1)}
=\frac{e_2^2}{qc_1^2}\frac{\prod\limits_{i=1,2}(1-a_i\lambda)(1-b_i\lambda)}{\prod\limits_{i=3,4}(1-a_i\lambda)(1-b_i\lambda)},
\nonumber\\
\frac{\big(\lambda\mu-\frac{e_2}{c_1}\big)\big(\o{\lambda}\mu-\frac{qe_2}{c_1}\big)}{(\lambda\mu-1)(\o{\lambda}\mu-1)}
=\frac{\prod\limits_{i=1,2}\big(\mu-\frac{a_ie_2}{c_1}\big)\big(\mu-\frac{b_ie_2}{c_1}\big)}{\prod\limits_{i=3,4}(\mu-a_i)(\mu-b_i)},\label{eq:E7T2_ev}
\end{gather}
where $e_2=\frac{a_3a_4d_1}{b_1b_2}$. Here eight singular points $(\lambda_s,\mu_s)$ $(s=1,\ldots,8)$ in coordinates $(\lambda,\mu)$ are on two curves $\lambda\mu=1$ and $\lambda\mu=\frac{e_2}{c_1}$ as follows
\begin{gather}\label{eq:E7T2_8p}
\left(\frac{1}{a_i}, \frac{a_ie_2}{c_1}\right)_{i=1,2},\quad \left(\frac{1}{b_i}, \frac{b_ie_2}{c_1}\right)_{i=1,2}, \quad \left(\frac{1}{a_i}, a_i\right)_{i=3,4},\quad \left(\frac{1}{b_i}, b_i\right)_{i=3,4}.
\end{gather}
The following scalar Lax equations
\begin{gather}
L_1(x)=\frac{A(x)}{1-\lambda x}\left[y(qx)-\frac{e_2\prod\limits_{i=1,2}(x-b_i)(\mu-x)}{\prod\limits_{i=3,4}(x-a_i)\big(\mu-\frac{e_2x}{c_1}\big)}y(x)\right]
\nonumber\\ \hphantom{L_1(x)=}{}
+\frac{q^2c_1^2B\big(\frac{x}{q}\big)}{1-\lambda\frac{x}{q}}\left[y\left(\frac{x}{q}\right)-\frac{\prod\limits_{i=3,4} \big(\frac{x}{q}-a_i\big)\big(\mu-\frac{e_2x}{qc_1}\big)}{e_2\prod\limits_{i=1,2}\big(\frac{x}{q}-b_i\big)\big(\mu-\frac{x}{q}\big)}y(x)\right]
\nonumber\\ \hphantom{L_1(x)=}
{}+\frac{c_1\big(1-\frac{c_1}{e_2}\big)x^2}{\mu}\left[\frac{\prod\limits_{i=3,4}(\mu-a_i)(\mu-b_i)}{(\lambda\mu-1)\big(\mu-\frac{x}{q}\big)} -\frac{\prod\limits_{i=1,2}\big(\mu-\frac{a_ie_2}{c_1}\big)\big(\mu-\frac{b_ie_2}{c_1}\big)}{\big(\lambda\mu-\frac{e_2}{c_1}\big)
\big(\mu-\frac{e_2x}{c_1}\big)}\right]y(x),
\label{eq:E7T2_L}\\
\nonumber L_2(x)=w_2x(1-\lambda x)\o{y}(x)-\prod_{i=3}^4(x-a_i)\left(1-\frac{e_2x}{c_1\mu}\right)y(qx)+e_2\prod_{i=1}^2(x-b_i)\left(1-\frac{x}{\mu}\right)y(z)
\end{gather}
are equivalent to those in \cite{KNY17, Nagao15,Yamada11} up to a gauge transformation of $y(x)$. Here $A$ and $B$ are as in (\ref{eq:E7T1_ABV}) and $w_2$ is a gauge freedom (as mentioned in Section~\ref{subsec:Lax_E7T1}). Then the compatibility of the~$L_1$ and~$L_2$ equations (\ref{eq:E7T2_L}) is equivalent to the system (\ref{eq:E7T2_ev}). Next setting the expres\-sion~$L_1^*$~by
\begin{gather}\label{eq:E7T2_L1*}
L_1^*(x)=(1-\lambda x)\left(1-\lambda\frac{x}{q}\right)(\lambda\mu-1)\left(\lambda\mu-\frac{e_2}{c_1}\right)L_1(x),
\end{gather}
we have:
\begin{Proposition}\label{prop:E7T2_L1*}
The $L_1^*$ equation (\ref{eq:E7T2_L1*}) has the following characterization:
\begin{enumerate}\itemsep=0pt
\item[$(i)$] The expression $L_1^*$ is a polynomial of bi-degree $(3,2)$ in the coordinates $(\lambda,\mu) \in \mathbb{P}^1 \times \mathbb{P}^1$.
\item[$(ii)$] As a polynomial, the expression $L_1^*$ vanishes at the following $12$ points $(\lambda_s, \mu_s) \in \mathbb{P}^1 \times \mathbb{P}^1$ $(s=1,\ldots,12)$:
\begin{gather}
 \left(\frac{1}{a_i}, \frac{a_ie_2}{c_1}\right)_{i=1,2},\quad \left(\frac{1}{b_i}, \frac{b_ie_2}{c_1}\right)_{i=1,2}, \quad \left(\frac{1}{a_i}, a_i\right)_{i=3,4},\quad \left(\frac{1}{b_i}, b_i\right)_{i=3,4}, \nonumber\\
\left(\frac{1}{x}, x\right),\quad \left(\frac{q}{x}, \frac{e_2x}{qc_1}\right),\quad \left(\frac{1}{x}, \mu_{\frac{1}{x}}\right), \quad \left(\frac{q}{x}, \mu_{\frac{q}{x}}\right),\label{eq:E7T2_12p}
\end{gather}
where the first $8$ points are as in \eqref{eq:E7T2_8p} and $\mu_u$ is given by
\begin{gather*}
\frac{Y(\frac{q}{u})}{Y(\frac{1}{u})}=:\frac{e_2\prod\limits_{i=1,2}\big(\frac{1}{u}-b_i\big) \big(\mu_u-\frac{1}{u}\big)}{\prod\limits_{i=3,4}\big(\frac{1}{u}-a_i\big)\big(\mu_u-\frac{e_2}{c_1u}\big)}, \qquad u=\frac{1}{x}, \frac{q}{x}.
\end{gather*}
\end{enumerate}
Conversely the equation $L_1^*$ is uniquely characterized by these properties $(i)$ and $(ii)$.
\end{Proposition}

\begin{proof}This proof is the similar as for Proposition \ref{prop:E7T1_L1*}.
\end{proof}

\subsection[Correspondence between two $L_1$ equations]{Correspondence between two $\boldsymbol{L_1}$ equations}\label{subsec:L1_Corres}

Thanks to Propositions \ref{prop:E7T1_L1*} and \ref{prop:E7T2_L1*}, we have:

\begin{Theorem}\label{thm:L1_Corres}
The $L_1$ equation \eqref{eq:E7T1_L1} is equivalent to Yamada's $L_1$ equation \eqref{eq:E7T2_L} under a~relation
\begin{gather}\label{eq:L1_Corres}
 f=\lambda, \qquad \frac{\big(\lambda\mu-\frac{e_2}{c_1}\big)\big(e_1f^2+e_1fg+c_1\big)}{e_2(1-b_2\lambda)(\lambda\mu-1)(1-a_2f)}=1.
\end{gather}
\end{Theorem}

\begin{proof}Comparing the last two points of (\ref{eq:E7T1_12p}) in Proposition \ref{prop:E7T1_L1*} with ones of (\ref{eq:E7T2_12p}) in Proposition \ref{prop:E7T2_L1*}, we obtain the transformation (\ref{eq:L1_Corres}) as a necessary condition to change the $L_1^*$ equation~(\ref{eq:E7T2_L1*}) into the $L_1^*$ equation (\ref{eq:E7T1_L1*}). Conversely, under the relation (\ref{eq:L1_Corres}), $\mu(f,g)$ is written as a rational function with the numerator and the denominator of bi-degree $(1,1)$ in $(f,g)$ respectively. Substituting the expression $\mu(f,g)$ (\ref{eq:L1_Corres}) into the $L_1^*$ equation (\ref{eq:E7T2_L1*}), it is shown that the algebraic curve $L_1^*=0$ (\ref{eq:E7T2_L1*}) of bi-degree $(3,2)$ in $(\lambda,\mu)$ changes to the algebraic curve of bi-degree~$(5,2)$ in~$(f,g)$. Furthermore, due to the relation~(\ref{eq:L1_Corres}), $12$~points~(\ref{eq:E7T2_12p}) are changed to $12$~points~(\ref{eq:E7T1_12p}) and $2$ lines $f=\frac{1}{a_2}$, $f=\frac{1}{b_2}$. Namely it turns out that the algebraic curve $L_1^*(\lambda,\mu)=0$~(\ref{eq:E7T2_L1*}) of bi-degree $(3,2)$ is changed to the algebraic curve $L_1^*(f,g) \times (1-a_2f)(1-b_2f)=0$ (\ref{eq:E7T1_L1*}) of bi-degree $(3,2) \times (2,0)$ in $(f,g)$. Hence, the relation~(\ref{eq:L1_Corres}) is proved to be the suf\/f\/icient condition that the $L_1$ equation (\ref{eq:E7T2_L}) corresponds with the $L_1$ equation~(\ref{eq:E7T1_L1}).
\end{proof}

We note that the system (\ref{eq:E7T1_ev}) has the Lax pair whose $L_1$ equation (\ref{eq:E7T1_L1}) is equivalent to that (\ref{eq:E7T2_L}) of the $q$-$E_7^{(1)}$ system (\ref{eq:E7T2_ev}) and clarify the relation (\ref{eq:L1_Corres}) of the dependent variables between the systems (\ref{eq:E7T1_ev}) and (\ref{eq:E7T2_ev}).

\section{Particular solutions}\label{sec:Ps}

In this section we recall the particular solutions for the system (\ref{eq:Gar_ev}) given in \cite{NY16}, and derive ones for the system (\ref{eq:E7T1_ev}) by the similar reduction as in Section~\ref{subsec:Lax_E7T1}.

\subsection[Case of the $q$-Garnier system]{Case of the $\boldsymbol{q}$-Garnier system}\label{subsec:Ps_Gar}
The contents are extracts from \cite[Section~5.2]{NY16}. For convenience we change the notations in Section~\ref{subsec:Lax_Gar} as follows
\begin{gather}\label{eq:Ps_sp}
\begin{bmatrix}
a_1&\ldots&a_N&a_{N+1}\\
b_1&\ldots&b_N&b_{N+1}\\
c_1&c_2&d_1&d_2
\end{bmatrix}
\mapsto
\begin{bmatrix}
\frac{1}{a_1}&\ldots&\frac{1}{a_N}&q^{m+n}\vspace{1mm}\\
\frac{1}{b_1}&\ldots&\frac{1}{b_N}&\frac{1}{q}\vspace{1mm}\\
cq^n\prod\limits_{1}^{N}\frac{b_i}{a_i}&q^m&c&1
\end{bmatrix},
\end{gather}
where $a_1, \ldots, a_N$, $b_1, \ldots, b_N$, $c$ $\in\C^{\times}$ and $m,n\in\Z_{\ge 0}$. Correspondingly we replace the nota\-tions~$A$,~$B$ etc.~(\ref{eq:Gar_ABFG}) by
\begin{gather}\label{eq:Ps_Gar_AB}
 A(x)=\prod_{i=1}^{N}(a_ix)_1, \qquad B(x)=\prod_{i=1}^{N}(b_ix)_1, \qquad A_1(x)=\frac{A(x)}{(a_1x)_1}, \qquad B_1(x)=\frac{B(x)}{(b_1x)_1}.
\end{gather}
We also replace the evolution direction (\ref{eq:E7T1_shift}) by
\begin{gather}\label{eq:Ps_Gar_shift}
T_1=T_{a_1}T_{b_1}.
\end{gather}
We show particular solutions in terms of the $\tau$ function
\begin{gather}\label{eq:Ps_Gar_tau}
 \tau_{m,n}=\det\left[{}_{N+1}\varphi_N \left(\substack{{b_1, \ldots, b_N ,q^{-(m+n)}}\\[3mm]{{a_1 ,\ldots, a_N}}},cq^{i+j+1}\right)\right]^n _{i,j=0},
\end{gather}
where the generalized $q$-hypergeometric function ${}_{N+1}\varphi_N$\footnote{In \cite{Suzuki15, Suzuki17} the particular solutions of a higher order $q$-Painlev\'e system was constructed in terms of the $q$-hypergeometric function ${}_{N+1}\varphi_N$ by T.~Suzuki.} \cite{GaR04} is def\/ined by (\ref{eq:HGF}). Then we have the following fact\footnote{For the proof of Proposition~\ref{prop:Ps_Gar}, see \cite[Section~5]{NY16}.}.
\begin{Proposition}\label{prop:Ps_Gar}
The polynomials $F(x)$ and $G(x)$ determined by
\begin{gather}
\frac{F\big(\frac{1}{a_i}\big)}{F\big(\frac{1}{b_j}\big)}=\alpha\frac{T_{a_i}(\tau_{m,n})T_{a_i}^{-1}(\tau_{m+1,n-1})}
{T_{b_j}^{-1}(\tau_{m,n})T_{b_j}(\tau_{m+1,n-1})}, \qquad i, j=1, \ldots, N,\nonumber\\
G\left(\frac{1}{a_i}\right)=\beta\frac{T_{a_i}(\o{\tau}_{m,n})T_{a_i}^{-1}(\tau_{m+1,n-1})}
{T_{a_1}(\tau_{m,n})T_{a_1}^{-1}(\o{\tau}_{m+1,n-1})}, \qquad i=2, \ldots, N,\nonumber\\
G\left(\frac{1}{b_i}\right)=\gamma\frac{T_{b_i}^{-1}(\tau_{m,n})T_{b_i}(\o{\tau}_{m+1,n-1})}
{T_{b_1}^{-1}(\o{\tau}_{m,n})T_{b_1}(\tau_{m+1,n-1})}, \qquad i=2, \ldots, N,\label{eq:Ps_Gar_FG}
\end{gather}
give particular solutions of a bi-rational equation
\begin{gather*}
G(x)\u{G}(x)=\frac{c\big(qx,\frac{x}{q^{m+n}}\big)_1 A_1 (x)B_1 (x)}{(a_1 x, b_1 x)_1}\qquad {\rm for} \quad F(x)=0,\\
F(x)\o{F}(x)=\left(qx,\frac{x}{q^{m+n}}\right)_1 A_1 (x)B_1 (x) \qquad {\rm for} \quad G(x)=0,\\
f_N \o{f}_N =\frac{q a_1b_1}{c}\Bigg(g_{N-1}-\frac{c\prod\limits_{i=2}^{N}(-b_i)}{a_1q^{m-1}}\Bigg)\Bigg(g_{N-1}-\frac{\prod\limits_{i=2}^{N}(-a_i)}{b_1q^n}\Bigg), \qquad
f_0\o{f}_0=\left(g_0,\frac{g_0}{c}\right)_1.
\end{gather*}
Here $\alpha$, $\beta$ and $\gamma$ are given by
\begin{gather*}
\alpha=-cq^{n-m}\frac{(a_i q^{m+n})_1 \big(\frac{a_i}{q}\big)_1^n \big(\frac{b_j}{q}\big)_1^n}{(a_i)_1^{n+1}(b_j)_1^n}\frac{B\big(\frac{1}{a_i}\big)}{A\big(\frac{1}{b_j}\big)},\\
\beta=c\frac{(b_1 , a_i q^{m+n})_1 \big(\frac{a_i}{q}\big)_1^n B_1 \big(\frac{1}{a_i}\big)}{a_1 q^m \big(\frac{b_1}{a_1}\big)_1 (a_i)_1^{n+1}},\qquad \gamma=\frac{(a_1)_1 (b_i)_1^n A_1\big(\frac{1}{b_i}\big)}{b_1 q^n \big(\frac{a_1}{b_1}\big)_1 \big(\frac{b_i}{q}\big)_1^n}.
\end{gather*}
\end{Proposition}

\subsection[Reduction to the $q$-Painlev\'e type system]{Reduction to the $\boldsymbol{q}$-Painlev\'e type system}\label{subsec:Ps_E7T1}
In this subsection we derive particular solutions for the Painlev\'e type equation (\ref{eq:E7T1_ev}). In a similar way as in Section~\ref{subsec:Lax_E7T1}, we consider the reduction from the particular case $N=3$ of the $q$-Garnier system~(\ref{eq:Gar_ev}). In order to do this, we impose a constraint of parameters
\begin{gather}\label{eq:Ps_E7T1_cons}
c=1, \qquad q^{m-n}\prod_{i=1}^3\frac{a_i}{b_i}=1,
\end{gather}
and the specialization (\ref{eq:E7T1_rn}). Then the tau function $\tau_{m,n}$ (\ref{eq:Ps_Gar_tau}) is reduced to the following function
\begin{gather*}
\tau_{m,n}=\det\left[{}_{4}\varphi_3 \left(\substack{{b_1, b_2, b_3 ,q^{-(m+n)}}\\[3mm]{{a_1 ,a_2, a_3}}},q^{i+j+1}\right)\right]^n _{i,j=0},
\end{gather*}
where the generalized $q$-hypergeometric function ${}_4\varphi_3$ is def\/ined by (\ref{eq:HGF}). As the case of a~reduction of Proposition~\ref{prop:Ps_Gar}, we have the following\footnote{As another derivation of the particular solutions~(\ref{eq:Ps_E7T1_fg}), see Section~\ref{subsec:Pade_Ps} (Pad\'e interpolation method).}.
\begin{Proposition}\label{prop:Ps_E7T1}
The particular values of $f$ and $g$ determined by
\begin{gather}
\frac{1-\frac{f}{a_i}}{1-\frac{f}{b_j}}=\alpha^{\prime}\frac{T_{a_i}(\tau_{m,n})T_{a_i}^{-1}(\tau_{m+1,n-1})} {T_{b_j}^{-1}(\tau_{m,n})T_{b_j}(\tau_{m+1,n-1})}, \qquad i,j=1,2,3,\nonumber\\
1+\frac{g}{a_i}+\frac{\kappa}{a_i^2}=\beta^{\prime}\frac{T_{a_i}(\o{\tau}_{m,n})T_{a_i}^{-1}(\tau_{m+1,n-1})}{T_{a_1}(\tau_{m,n}) T_{a_1}^{-1}(\o{\tau}_{m+1,n-1})}, \qquad i=2,3,\nonumber\\
1+\frac{g}{b_i}+\frac{\kappa}{b_i^2}=\gamma^{\prime}\frac{T_{b_i}^{-1}(\tau_{m,n})T_{b_i}(\o{\tau}_{m+1,n-1})}{T_{b_1}^{-1}(\o{\tau}_{m,n}) T_{b_1}(\tau_{m+1,n-1})}, \qquad i=2,3,\label{eq:Ps_E7T1_fg}
\end{gather}
give particular solutions of the following overdetermined bi-rational equation:
\begin{gather}
\big(f^2+gf+\kappa\big)\big(f^2+\u{g}f+q\kappa\big)=\frac{\big(f-\frac{1}{q^{m+n}}\big)(f-q)\prod\limits_{i=2,3}(f-a_i)(f-b_i)}{(f-a_1)(f-b_1)},\nonumber\\
\frac{x_1^2(1-fx_1)(1-\o{f}x_1)}{x_2^2(1-fx_2)(1-\o{f}x_2)}=\frac{\big(1-\frac{x_1}{q^{m+n}}\big)(1-qx_1)\prod\limits_{i=2,3}(1-a_ix_1)(1-b_ix_1)}
{\big(1-\frac{x_2}{q^{m+n}}\big)(1-qx_2)\prod\limits_{i=2,3}(1-a_ix_2)(1-b_ix_2)}.
\label{eq:Ps_E7T1_ev}\end{gather}
Here $\alpha^{\prime}$, $\beta^{\prime}$ and $\gamma^{\prime}$ are given by
\begin{gather*}
\alpha^{\prime}=-q^{n-m}\frac{a_i (a_iq^{m+n})_1\big(\frac{a_i}{q},\frac{b_j}{q}\big)_1^n}{b_j(a_i)_1^{n+1}(b_j)_1^n}
\frac{\prod\limits_{s=1}^3\big(\frac{b_s}{a_i}\big)_1}{\prod\limits_{s=1}^3\big(\frac{a_s}{b_j}\big)_1}, \\
\beta^{\prime}=\frac{\big(b_1,a_iq^{m+n},\frac{b_2}{a_i},\frac{b_3}{a_i}\big)_1\big(\frac{a_i}{q}\big)_1^n}{a_1q^m\big(\frac{b_1}{a_1}\big)_1(a_i)_1^{n+1}},\qquad
\gamma^{\prime}=\frac{\big(a_1,\frac{a_2}{b_i},\frac{a_3}{b_i}\big)_1(b_i)_1^n}{b_1q^n\big(\frac{a_1}{b_1}\big)_1\big(\frac{b_2}{q}\big)_1^n},
\end{gather*}
and the evolution direction is as in $T_1$ (\ref{eq:Ps_Gar_shift}) and $x=x_1, x_2$ are solutions of an equation $\varphi=0$:
\begin{gather}\label{eq:Ps_E7T1_phi}
\varphi(x)=1+gx+\kappa x^2,
\end{gather}
where $\kappa=\frac{a_2a_3}{b_1q^n}$.
\end{Proposition}

\begin{proof}
Substituting the conditions (\ref{eq:Ps_E7T1_cons}) and (\ref{eq:E7T1_rn}) into the particular solutions (\ref{eq:Ps_Gar_FG}),
we obtain~(\ref{eq:Ps_E7T1_fg}).
\end{proof}

\section{Conclusions}\label{sec:conc}

The main results of this paper are the following.
\begin{itemize}\itemsep=0pt
\item We showed in Proposition \ref{prop:BC} that the $q$-Painlev\'e type system (\ref{eq:E7T1_ev}) is the bi-rational transformation and is related to the novel realization (i.e., conf\/iguration) (\ref{eq:E7T1_8p}) of the symmetry/surface of type $E_7^{(1)}$/$A_1^{(1)}$. Then the system (\ref{eq:E7T1_ev}) turned out to be a variation of the $q$-$E_7^{(1)}$ system (\ref{eq:E7T2_ev}).
\item We obtained the Lax equations (\ref{eq:E7T1_L}) for the system (\ref{eq:E7T1_ev}) from the reduction of the particular case $N=3$ of the $q$-Garnier system (\ref{eq:Gar_ev}), and clarif\/ied the connection between the system (\ref{eq:E7T1_ev}) and the $q$-$E_7^{(1)}$ system (\ref{eq:E7T2_ev}) by comparing their $L_1$ equations in Theorem~\ref{thm:L1_Corres}.
\item In Proposition \ref{prop:Ps_E7T1} the determinant formulas of the particular solutions for the system~(\ref{eq:E7T1_ev}) was expressed in terms of the generalized $q$-hypergeometric function ${}_4\varphi_3$ through the si\-milar reduction of the particular case $N=3$.
\end{itemize}
Extending the results of this paper, we naturally have the following open problems. One may consider several variations of the $q$-$E_7$ system according to several deformation direction such as $T_1$ and $T_2$, and investigates a connection among these systems. We will carry out similar research on discrete Painlev\'e and Garnier systems \cite{NY17}. It seems to be interesting to study reductions of cases $N\geq4$ of the $q$-Garnier system.

\appendix

\section[From Pad\'e interpolation to $q$-Painlev\'e type system]{From Pad\'e interpolation to $\boldsymbol{q}$-Painlev\'e type system}\label{sec:Pade}

By using a Pad\'e interpolation problem with $q$-grid, in \cite{Nagao15} we derived the scalar Lax pair, the evolution equation and the particular solutions for the $q$-$E_7^{(1)}$ system. In this appendix, in a~similar manner as in \cite{Nagao15}, we directly derive the data of the $q$-Painlev\'e type system~(\ref{eq:E7T1_ev}) given in Sections~\ref{subsec:Lax_E7T1} and~\ref{subsec:Ps_E7T1}.
\subsection{Scalar Lax pair and evolution equation
}\label{subsec:Pade_Lax}
Suppose we have complex parameters $q$ ($|q| <1$), $a_1$, $a_2$, $a_3$, $b_1$, $b_2$ and $b_3 \in \C^{\times}$ with the constraint~(\ref{eq:Ps_E7T1_cons}). Then we consider a function
\begin{gather}\label{eq:Pade_Lax_psi}
\psi(x)=\prod_{i=1}^{3}\frac{(a_i x, b_i)_\infty}{(a_i, b_i x)_\infty}.
\end{gather}
Let $P(x)$ and $Q(x)$ be polynomials of degree $m$ and $n$ $\in \Z_{\ge 0}$ in $x$. Then we assume that the polynomials $P$ and $Q$ satisfy the following Pad\'e interpolation condition:
\begin{gather}\label{eq:Pade_Lax_pade}
\psi(x_s)=\frac{P(x_s)}{Q(x_s)}, \qquad x_s=q^s, \quad s=0, 1, \ldots, m+n.
\end{gather}
The common normalizations of the polynomials $P$ and $Q$ in $x$ are f\/ixed as $P(0)=1$. The parameter shift operator is given by $T_1$ (\ref{eq:Ps_Gar_shift}). Consider two linear relations: $L_2=0$ among $y(x)$, $y(qx)$, $\o{y}(x)$ and $L_3=0$ among $y(x)$, $\o{y}(x)$, $\o{y}\big(\frac{x}{q}\big)$ satisf\/ied by the functions $y=P$ and $y=\psi Q$. Then we have:
\begin{Proposition}
The linear relations $L_2$ and $L_3$\footnote{$L_2 = 0$ and $L_3 = 0$ \eqref{eq:Pade_Lax_L2L3} can be derived by substituting \eqref{eq:Ps_sp}, \eqref{eq:Ps_E7T1_cons} and $w_1\to -\frac{q^{2m}C_0}{a_2a_3}$, $\overline{w}_1\to -\frac{q^{2m}C_1}{b_2b_3}$
into $L_2 = 0$ and $L_3 = 0$ \eqref{eq:E7T1_L}.} can be expressed as follows
\begin{gather} L_2(x)=C_0x(fx)_1\o{y}(x)- \left(\frac{x}{q^{m+n}}\right)_1A_1(x)y(qx)+(b_1x)_1\varphi(x)y(x)=0,\nonumber\\
L_3(z)=C_1\frac{x}{q}\left(\frac{\o{f}x}{q}\right)_1y(x)+(a_1x)_1\varphi\left(\frac{x}{q}\right)\o{y}(x)-(x)_1B_1
\left(\frac{x}{q}\right)\o{y}\left(\frac{x}{q}\right)=0,\label{eq:Pade_Lax_L2L3}
\end{gather}
where $\varphi$ is given by \eqref{eq:Ps_E7T1_phi} and $A_1$, $B_1$ are the same as the case $N=3$ of \eqref{eq:Ps_Gar_AB}. Here $f,g,C_0,C_1\in\P^1$ are constants depending on parameters $a_i, b_j \in\C^{\times}$, $m, n \in\Z_{\ge 0}$.
\end{Proposition}

\begin{proof}
By the def\/inition of the relations $L_2=0$ and $L_3=0$, they can be written as
\begin{gather}
L_2(x)\propto
\begin{vmatrix}
y(x) & y(qx) & \o{y}(x) \\
{\bf y}(x) & {\bf y}(qx) & \o{\bf y}(x)
\end{vmatrix}
=D_1(x) \o{y}(x)-D_2(x)y(qx)+D_3(x)y(x)=0,\nonumber\\ L_3(x)\propto
\begin{vmatrix}
y(x) & \o{y}(x) & \o{y}\big(\frac{x}{q}\big) \\
{\bf y}(x) & \o{\bf y}(x) & \o{\bf y}\big(\frac{x}{q}\big)
\end{vmatrix}
=\o{D}_1\left(\frac{x}{q}\right)y(x)+D_3\left(\frac{x}{q}\right) \o{y}(x)-D_2(x)\o{y}\left(\frac{x}{q}\right)=0,\label{eq:Pade_Lax_mat}
\end{gather}
where ${\bf y}(x)=\left[\begin{matrix}P(x)\\ \psi(x)Q(x)\end{matrix}\right]$ and Casorati determinants
\begin{gather}\label{eq:Pade_Lax_D}
D_1(x)=|{\bf y}(x),{\bf y}(qx)|, \qquad D_2(x)=|{\bf y}(x),{\o{\bf y}}(x)|, \qquad D_3(x)=|{\bf y}(qx),\o{{\bf y}}(x)|.
\end{gather}
Taking note of the relations
\begin{gather}\label{eq:Pade_Lax_ratio}
\frac{\psi(qx)}{\psi(x)}=\frac{B(x)}{A(x)},
\qquad \frac{\o{\psi}(x)}{\psi(x)}=\frac{(a_1,b_1x)_1}{(a_1x,b_1)_1},
\end{gather}
where $A$ and $B$ are the same as the case $N=3$ of (\ref{eq:Ps_Gar_AB}), we rewrite the Casorati determinants~(\ref{eq:Pade_Lax_D}) into the following determinants
\begin{gather} D_1(x)=\frac{\psi(x)}{A(x)}R_1(x)
=:\frac{\psi(x)}{A(x)}\prod_{i=0}^{m+n-1}\left(\frac{x}{q^i}\right)_1c_0 x(fx)_1,\nonumber\\
 D_2(x)=\frac{\psi(x)}{(a_1x,b_1)_1}R_2(x)=:\frac{\psi(x)}{(a_1x,b_1)_1}\prod_{i=0}^{m+n}\left(\frac{x}{q^i}\right)_1c_0^{\prime},\nonumber\\
 D_3(x)=\frac{\psi(x)}{A(x)}R_3(x)=:\frac{\psi(x)}{A(x)(b_1)_1}\prod_{i=0}^{m+n-1}\left(\frac{x}{q^i}\right)_1c_0^{\prime}(b_1x)_1\varphi(x),\label{eq:Pade_Lax_DD}
\end{gather}
where
\begin{gather}
R_1(x)=B(x)P(x)Q(qx)-A(x)P(qx)Q(x), \nonumber\\
R_2(x)=(a_1,b_1x)_1P(x)\o{Q}(x)-(a_1x,b_1)_1\o{P}(x)Q(x), \nonumber\\
R_3(x)=(a_1,b_1x)_1A_1(x)P(qx)\o{Q}(x)-(b_1)_1B(x)\o{P}(x)Q(qx).\label{eq:Pade_Lax_R}
\end{gather}
Here $c_0$ and $c_0^{\prime}$ are some constants depending on the parameters $a_i$, $b_j$, $m$ and $n$.
Computing Taylor expansions at $x=0$ and $x=\infty$ in the expressions $R_2(x)$ and $R_3(x)$(\ref{eq:Pade_Lax_R}), we determine~$\varphi$ by~(\ref{eq:Ps_E7T1_phi}). As a result, we obtain the desired relations $L_2$ and $L_3$ (\ref{eq:Pade_Lax_L2L3}) where
$C_0=\frac{c_0}{c_0^{\prime}}$ and $C_1=\frac{(a_1)_1\o{c}_0}{c_0^{\prime}}$.
\end{proof}

Next we have:
\begin{Proposition}\label{prop:Pade_Lax_ev}
The constants $f$ and $g$ satisfy the $q$-Painlev\'e type system \eqref{eq:Ps_E7T1_ev}, and they play the role of dependent variables for~\eqref{eq:Ps_E7T1_ev}.
\end{Proposition}
\begin{proof}
The compatibility of the relations (\ref{eq:Pade_Lax_L2L3}) gives the system (\ref{eq:Ps_E7T1_ev}).
\end{proof}

\subsection{Particular solution}\label{subsec:Pade_Ps}
We construct particular solutions of the $q$-Painlev\'e type system (\ref{eq:E7T1_ev}) given in terms of the $q$-hypergeometric function ${}_{4}\varphi_3$ in Section~\ref{subsec:Ps_E7T1}. We derive the explicit forms~(\ref{eq:Ps_E7T1_fg}) of variables $\{f, g\}$ appearing in the Casorati determinants $D_1$ and $D_3$ (\ref{eq:Pade_Lax_DD}). They are interpreted as the particular solutions for the system (\ref{eq:Ps_E7T1_ev}), due to Proposition~\ref{prop:Pade_Lax_ev}.

\begin{Proposition}[\cite{Jacobi}, see also \cite{Ikawa13, Nagao15, NY16}]\label{prop:Pade_Ps_Jacobi}
For a given sequence $\psi_s$, the polynomials $P(x)$ and~$Q(x)$ of degree $m$ and $n$ satisfying a Pad\'e interpolation problem
\begin{gather} \label{eq:Pade_Ps_pade}
\psi_s=\frac{P (x_s)}{Q (x_s)}, \qquad s=0, 1, \dots, m+n,
\end{gather}
is given as the following determinant expressions:
\begin{gather}
 P(x)=\mathcal{F}(x)\det\left[\sum^{m+n} _{s=0}u_s \frac{x_s^{i+j}}{x-x_s }\right]^n _{i,j=0}, \nonumber\\
Q(x)=\det\left[\sum^{m+n} _{s=0}u_s x_s^{i+j}(x-x_s )\right]^{n-1} _{i,j=0},\label{eq:Pade_Ps_Jacobi}
\end{gather}
where $u_s=\frac{\psi_s}{\mathcal{F}^{\prime}(x_s)}$ and $\mathcal{F}(x)=\prod\limits_{i=0}^{m+n}(x-x_i )$.
\end{Proposition}

\begin{Proposition}[\cite{Ikawa13, Nagao15, NY16}]\label{prop:Pade_Ps_qJacobi}
In the $q$-grid case of the problem \eqref{eq:Pade_Ps_pade} $($i.e., interpolation points $x_s=q^s)$, the formula \eqref{eq:Pade_Ps_Jacobi} takes the following form:
\begin{gather}
P(x)=\frac{\mathcal{F}(x)}{(q)_{m+n}^{n+1}}\det\left[\sum^{m+n} _{s=0}\psi_s\frac{(q^{-(m+n)})_s }{(q)_s}\frac{q^{s(i+j+1)}}{x-q^s }\right]^n _{i,j=0},\nonumber\\
Q(x)=\frac{1}{(q)_{m+n}^{n}}\det\left[\sum^{m+n} _{s=0}\psi_s\frac{(q^{-(m+n)})_s }{(q)_s }q^{s(i+j+1)}(x-q^s)\right]^{n-1} _{i,j=0}.
\label{eq:Pade_Ps_qJacobi}
\end{gather}
\end{Proposition}

\begin{proof}
Substituting the expressions
\begin{gather}\label{eq:Pade_Ps_fprime}
\mathcal{F}(x)=\prod_{s=0}^{m+n}\big(x-q^s \big), \qquad \mathcal{F}^{\prime}(x_s )
=\frac{(q)_s (q)_{m+n}}{q^{s}(q^{-(m+n)})_s},
\end{gather}
into the formula (\ref{eq:Pade_Ps_Jacobi}), then we obtain the desired form (\ref{eq:Pade_Ps_qJacobi}).
\end{proof}

\begin{Remark}\label{rem:Pade_CN}
The normalization of the polynomials $P(x)$ and $Q(x)$ expressed in the formulas~(\ref{eq:Pade_Ps_Jacobi}) and~(\ref{eq:Pade_Ps_qJacobi}) dif\/fer from the convention $P(0)=1$ as f\/ixed beneath the interpolation condition (\ref{eq:Pade_Lax_pade}). This dif\/ference does not inf\/luence the result in Proposition~\ref{prop:Pade_Ps_fg}, because the common normalization factors of $P$ and $Q$ cancel in (\ref{eq:Pade_Ps_f}) and (\ref{eq:Pade_Ps_g}).
\end{Remark}

\begin{Proposition}\label{prop:Pade_Ps_PQ}
The polynomials $P(x)$ and $Q(x)$ defined in Section~{\rm \ref{subsec:Pade_Lax}} have the following particular values:
\begin{gather}
P\left(\frac{1}{a_s}\right)=\frac{(a_s)_{m+n+1}}{a_s^m (a_s)_1^{n+1}(q)_{m+n}^{n+1}}T_{a_s}(\tau_{m,n}),\qquad Q\left(\frac{q}{a_s}\right)=\frac{q^n \big(\frac{a_s}{q}\big)_1^n}{a_s^n (q)_{m+n}^n}T_{a_s}^{-1}(\tau_{m+1,n-1}),\nonumber\\
P\left(\frac{q}{b_s}\right)=\frac{q^m \big(\frac{b_s}{q}\big)_{m+n+1}}{b_s^m \big(\frac{b_s}{q}\big)_1^{n+1}(q)_{m+n}^{n+1}}T_{b_s}^{-1}(\tau_{m,n}),\qquad Q\left(\frac{1}{b_s}\right)=\frac{(b_s)_1^n}{b_s^n (q)_{m+n}^n}T_{b_s}(\tau_{m+1,n-1}),\label{eq:Pade_Ps_PQ}
\end{gather}
for $s=1,2, 3$. Here $\tau_{m,n}$ is defined by \eqref{eq:Ps_Gar_tau}.
\end{Proposition}

\begin{proof}
This proof follows from the formula (\ref{eq:Pade_Ps_qJacobi}) and the sequence $\psi_s=\psi(q^s)=\prod\limits_{i=1}^3\frac{(b_i)_s}{(a_i)_s}$.
\end{proof}

\begin{Proposition}\label{prop:Pade_Ps_fg}
The particular values of $f$ and $g$ determined by \eqref{eq:Ps_E7T1_fg} give particular solutions of the system \eqref{eq:Ps_E7T1_ev}.
\end{Proposition}

\begin{proof}
From the f\/irst equation of (\ref{eq:Pade_Lax_DD}), we have
\begin{gather}\label{eq:Pade_Ps_f}
\frac{1-\frac{f}{a_i}}{1-\frac{f}{b_j}}=-\frac{a_i}{b_j}\prod_{s=0}^{m+n-1}\frac{\big(\frac{1}{b_j q^s}\big)_1}{\big(\frac{1}{a_i q^s}\big)_1}\frac{B\big(\frac{1}{a_i}\big)}{A\big(\frac{1}{b_j}\big)}\frac{P\big(\frac{1}{a_i}\big)Q\big(\frac{q}{a_i}\big)}
{P\big(\frac{q}{b_j}\big)Q\big(\frac{1}{b_j}\big)}, \qquad i, j=1,2,3,
\end{gather}
where $A$ and $B$ are as in Appendix~\ref{subsec:Pade_Lax}. From the second and third equations of (\ref{eq:Pade_Lax_DD}), we have
\begin{gather}
 1+\frac{g}{a_i}+\frac{\kappa}{a_i^2}=-\frac{\prod\limits_{s=0}^{m+n}\big(\frac{1}{a_1 q^s}\big)_1}{\prod\limits_{s=0}^{m+n-1}\big(\frac{1}{a_i q^s}\big)_1}\frac{(b_1)_1 B_1\big(\frac{1}{a_i}\big)}{\big(a_1 ,\frac{b_1}{a_1}\big)_1}\frac{\o{P}\big(\frac{1}{a_i}\big)Q\big(\frac{q}{a_i}\big)}
{P\big(\frac{1}{a_1}\big)\o{Q}\big(\frac{1}{a_1}\big)}, \qquad i=2, 3, \nonumber\\
 1+\frac{g}{b_i}+\frac{\kappa}{b_i^2}=-\frac{\prod\limits_{s=0}^{m+n}\big(\frac{1}{b_1 q^s}\big)_1}{\prod\limits_{s=0}^{m+n-1}\big(\frac{1}{b_i q^s}\big)_1}\frac{(a_1)_1A_1\big(\frac{1}{b_i}\big)}{\big(\frac{a_1}{b_1},b_1\big)_1}\frac{P\big(\frac{q}{b_i}\big)\o{Q}\big(\frac{1}{b_i}\big)}
{\o{P}\big(\frac{1}{b_1}\big)Q\big(\frac{1}{b_1}\big)} ,\qquad i=2, 3,\label{eq:Pade_Ps_g}
\end{gather}
where $A_1$ and $B_1$ are as in Appendix~\ref{subsec:Pade_Lax}. Substituting the particular values (\ref{eq:Pade_Ps_PQ}) into the expressions (\ref{eq:Pade_Ps_f}) and (\ref{eq:Pade_Ps_g}) respectively, we obtain the desired particular solutions~(\ref{eq:Ps_E7T1_fg}).
\end{proof}

\subsection*{Acknowledgements}

The author shall be thankful to Professor Yasuhiko Yamada for valuable discussions. The author is also grateful to the referees for stimulating comments. This work was partially supported by Expenses Revitalizing Education and Research of Akashi College (0217030).

\pdfbookmark[1]{References}{ref}
\LastPageEnding


\begin{thebibliography}{99}
\footnotesize\itemsep=0pt

\bibitem{DST13}
Dzhamay A., Sakai H., Takenawa T., Discrete {S}chlesinger transformations,
 their {H}amiltonian formulation, and dif\/ference {P}ainlev\'e equations,
 \href{https://arxiv.org/abs/1302.2972}{arXiv:1302.2972}.

\bibitem{DT14}
Dzhamay A., Takenawa T., Geometric analysis of reductions from {S}chlesinger
 transformations to dif\/ference {P}ainlev\'e equations, in Algebraic and
 Analytic Aspects of Integrable Systems and {P}ainlev\'e Equations,
 \href{https://doi.org/10.1090/conm/651/13044}{\textit{Contemp. Math.}}, Vol.~651, Amer. Math. Soc., Providence, RI, 2015,
 87--124, \href{https://arxiv.org/abs/1408.3778}{arXiv:1408.3778}.

\bibitem{GaR04}
Gasper G., Rahman M., Basic hypergeometric series, \href{https://doi.org/10.1017/CBO9780511526251}{\textit{Encyclopedia of
 Mathematics and its Applications}}, Vol.~96, 2nd ed., Cambridge University
 Press, Cambridge, 2004.

\bibitem{GORS98}
Grammaticos B., Ohta Y., Ramani A., Sakai H., Degeneration through coalescence
 of the {$q$}-{P}ainlev\'e {VI} equation, \href{https://doi.org/10.1088/0305-4470/31/15/018}{\textit{J.~Phys.~A: Math. Gen.}}
 \textbf{31} (1998), 3545--3558.

\bibitem{GR99}
Grammaticos B., Ramani A., On a novel {$q$}-discrete analogue of the
 {P}ainlev\'e~{VI} equation, \href{https://doi.org/10.1016/S0375-9601(99)00296-0}{\textit{Phys. Lett.~A}} \textbf{257} (1999),
 288--292.

\bibitem{Ikawa13}
Ikawa Y., Hypergeometric solutions for the {$q$}-{P}ainlev\'e equation of type
 {$E^{(1)}_6$} by the {P}ad\'e method, \href{https://doi.org/10.1007/s11005-013-0610-0}{\textit{Lett. Math. Phys.}} \textbf{103}
 (2013), 743--763, \href{https://arxiv.org/abs/1207.6446}{arXiv:1207.6446}.

\bibitem{Jacobi}
Jacobi C.G.J., \"Uber die {D}arstellung einer {R}eihe gegebner {W}erthe durch
 eine gebrochne rationale {F}unction, \href{https://doi.org/10.1515/crll.1846.30.127}{\textit{J.~Reine Angew. Math.}}
 \textbf{30} (1846), 127--156.

\bibitem{JM81-1}
Jimbo M., Miwa T., Ueno K., Monodromy preserving deformation of linear ordinary
 dif\/ferential equations with rational coef\/f\/icients. {I}.~{G}eneral theory and
 {$\tau $}-function, \href{https://doi.org/10.1016/0167-2789(81)90013-0}{\textit{Phys.~D}} \textbf{2} (1981), 306--352.

\bibitem{JM81-2}
Jimbo M., Miwa T., Monodromy preserving deformation of linear ordinary
 dif\/ferential equations with rational coef\/f\/icients.~{II}, \href{https://doi.org/10.1016/0167-2789(81)90021-X}{\textit{Phys.~D}}
 \textbf{2} (1981), 407--448.

\bibitem{JM81-3}
Jimbo M., Miwa T., Monodromy preserving deformation of linear ordinary
 dif\/ferential equations with rational coef\/f\/icients.~{III}, \href{https://doi.org/10.1016/0167-2789(81)90003-8}{\textit{Phys.~D}}
 \textbf{4} (1981), 26--46.

\bibitem{JS96}
Jimbo M., Sakai H., A {$q$}-analog of the sixth {P}ainlev\'e equation,
 \href{https://doi.org/10.1007/BF00398316}{\textit{Lett. Math. Phys.}} \textbf{38} (1996), 145--154.

\bibitem{KMNOY03}
Kajiwara K., Masuda T., Noumi M., Ohta Y., Yamada Y., {${}_{10}E_9$} solution
 to the elliptic {P}ainlev\'e equation, \href{https://doi.org/10.1088/0305-4470/36/17/102}{\textit{J.~Phys.~A: Math. Gen.}}
 \textbf{36} (2003), L263--L272, \href{https://arxiv.org/abs/nlin.SI/0303032}{nlin.SI/0303032}.

\bibitem{KMNOY04}
Kajiwara K., Masuda T., Noumi M., Ohta Y., Yamada Y., Hypergeometric solutions
 to the {$q$}-{P}ainlev\'e equations, \href{https://doi.org/10.1155/S1073792804140919}{\textit{Int. Math. Res. Not.}}
 \textbf{2004} (2004), 2497--2521, \href{https://arxiv.org/abs/nlin.SI/0403036}{nlin.SI/0403036}.

\bibitem{KMNOY05}
Kajiwara K., Masuda T., Noumi M., Ohta Y., Yamada Y., Construction of
 hypergeometric solutions to the {$q$}-{P}ainlev\'e equations, \href{https://doi.org/10.1155/IMRN.2005.1439}{\textit{Int.
 Math. Res. Not.}} \textbf{2005} (2005), 1441--1463, \href{https://arxiv.org/abs/nlin.SI/0501051}{nlin.SI/0501051}.

\bibitem{KNY02-2}
Kajiwara K., Noumi M., Yamada Y., {$q$}-{P}ainlev\'e systems arising from
 {$q$}-{KP} hierarchy, \href{https://doi.org/10.1023/A:1022216308475}{\textit{Lett. Math. Phys.}} \textbf{62} (2002),
 259--268, \href{https://arxiv.org/abs/nlin.SI/0112045}{nlin.SI/0112045}.

\bibitem{KNY17}
Kajiwara K., Noumi M., Yamada Y., Geometric aspects of {P}ainlev\'e equations,
 \href{https://doi.org/10.1088/1751-8121/50/7/073001}{\textit{J.~Phys.~A: Math. Theor.}} \textbf{50} (2017), 073001, 164~pages,
 \href{https://arxiv.org/abs/1509.08186}{arXiv:1509.08186}.

\bibitem{Knuth92}
Knuth D.E., Two notes on notation, \href{https://doi.org/10.2307/2325085}{\textit{Amer. Math. Monthly}} \textbf{99}
 (1992), 403--422.

\bibitem{Mano12}
Mano T., Determinant formula for solutions of the {G}arnier system and {P}ad\'e
 approximation, \href{https://doi.org/10.1088/1751-8113/45/13/135206}{\textit{J.~Phys.~A: Math. Theor.}} \textbf{45} (2012), 135206,
 14~pages.

\bibitem{MT14}
Mano T., Tsuda T., Two approximation problems by {H}ermite and the
 {S}chlesinger transformations, in Novel Development of Nonlinear Discrete
 Integrable Systems, \textit{RIMS K\^oky\^uroku Bessatsu}, Vol.~B47, Res.
 Inst. Math. Sci. (RIMS), Kyoto, 2014, 77--86.

\bibitem{MT17}
Mano T., Tsuda T., Hermite--{P}ad\'e approximation, isomonodromic deformation
 and hypergeometric integral, \href{https://doi.org/10.1007/s00209-016-1713-y}{\textit{Math.~Z.}} \textbf{285} (2017), 397--431,
 \href{https://arxiv.org/abs/1502.06695}{arXiv:1502.06695}.

\bibitem{Masuda09}
Masuda T., Hypergeometric {$\tau$}-functions of the {$q$}-{P}ainlev\'e system
 of type {$E_7^{(1)}$}, \href{https://doi.org/10.3842/SIGMA.2009.035}{\textit{SIGMA}} \textbf{5} (2009), 035, 30~pages,
 \href{https://arxiv.org/abs/0903.4102}{arXiv:0903.4102}.

\bibitem{Masuda11}
Masuda T., Hypergeometric {$\tau$}-functions of the {$q$}-{P}ainlev\'e system
 of type {$E^{(1)}_8$}, \href{https://doi.org/10.1007/s11139-010-9262-1}{\textit{Ramanujan~J.}} \textbf{24} (2011), 1--31.

\bibitem{Murata09}
Murata M., Lax forms of the {$q$}-{P}ainlev\'e equations, \href{https://doi.org/10.1088/1751-8113/42/11/115201}{\textit{J.~Phys.~A:
 Math. Theor.}} \textbf{42} (2009), 115201, 17~pages, \href{https://arxiv.org/abs/0810.0058}{arXiv:0810.0058}.

\bibitem{Nagao15}
Nagao H., The {P}ad\'e interpolation method applied to {$q$}-{P}ainlev\'e
 equations, \href{https://doi.org/10.1007/s11005-015-0749-y}{\textit{Lett. Math. Phys.}} \textbf{105} (2015), 503--521,
 \href{https://arxiv.org/abs/1409.3932}{arXiv:1409.3932}.

\bibitem{Nagao17-1}
Nagao H., The {P}ad\'e interpolation method applied to {$q$}-{P}ainlev\'e
 equations~{II} (dif\/ferential grid version), \href{https://doi.org/10.1007/s11005-016-0899-6}{\textit{Lett. Math. Phys.}}
 \textbf{107} (2017), 107--127, \href{https://arxiv.org/abs/1509.05892}{arXiv:1509.05892}.

\bibitem{Nagao16}
Nagao H., Lax pairs for additive dif\/ference {P}ainlev\'e equations,
\href{https://arxiv.org/abs/1604.02530}{arXiv:1604.02530}.

\bibitem{Nagao17-2}
Nagao H., Hypergeometric special solutions for $d$-{P}ainlev\'e equations,
 \href{https://arxiv.org/abs/1706.10101}{arXiv:1706.10101}.

\bibitem{NY16}
Nagao H., Yamada Y., Study of $q$-{G}arnier system by {P}ad\'e method,
 \href{https://arxiv.org/abs/1601.01099}{arXiv:1601.01099}.

\bibitem{NY17}
Nagao H., Yamada Y., Variations of $q$-{G}arnier system, \href{https://arxiv.org/abs/1710.03998}{arXiv:1710.03998}.


\bibitem{NTY13}
Noumi M., Tsujimoto S., Yamada Y., Pad\'e interpolation for elliptic
 {P}ainlev\'e equation, in Symmetries, Integrable Systems and Representations,
 \href{https://doi.org/10.1007/978-1-4471-4863-0_18}{\textit{Springer Proc. Math. Stat.}}, Vol.~40, Springer, Heidelberg, 2013,
 463--482, \href{https://arxiv.org/abs/1204.0294}{arXiv:1204.0294}.

\bibitem{ORG01}
Ohta Y., Ramani A., Grammaticos B., An af\/f\/ine {W}eyl group approach to the
 eight-parameter discrete {P}ainlev\'e equation, \href{https://doi.org/10.1088/0305-4470/34/48/316}{\textit{J.~Phys.~A: Math.
 Gen.}} \textbf{34} (2001), 10523--10532.

\bibitem{OKSO06}
Ohyama Y., Kawamuko H., Sakai H., Okamoto K., Studies on the {P}ainlev\'e
 equations. {V}.~{T}hird {P}ainlev\'e equations of special type {$P_{\rm
 III}(D_7)$} and {$P_{\rm III}(D_8)$}, \textit{J.~Math. Sci. Univ. Tokyo}
 \textbf{13} (2006), 145--204.

\bibitem{OO06}
Ohyama Y., Okumura S., A coalescent diagram of the {P}ainlev\'e equations from
 the viewpoint of isomo\-nodromic deformations, \href{https://doi.org/10.1088/0305-4470/39/39/S08}{\textit{J.~Phys.~A: Math. Gen.}}
 \textbf{39} (2006), 12129--12151, \href{https://arxiv.org/abs/math.CA/0601614}{math.CA/0601614}.

\bibitem{Okamoto79}
Okamoto K., Sur les feuilletages associ\'es aux \'equations du second ordre \`a
 points critiques f\/ixes de {P}.~{P}ainlev\'e, \textit{Japan.~J. Math. (N.S.)}
 \textbf{5} (1979), 1--79.

\bibitem{OR16-1}
Ormerod C.M., Rains E.M., Commutation relations and discrete {G}arnier systems,
 \href{https://doi.org/10.3842/SIGMA.2016.110}{\textit{SIGMA}} \textbf{12} (2016), 110, 50~pages, \href{https://arxiv.org/abs/1601.06179}{arXiv:1601.06179}.

\bibitem{OR16-2}
Ormerod C.M., Rains E.M., An elliptic {G}arnier system, \href{https://doi.org/10.1007/s00220-017-2934-6}{\textit{Comm. Math.
 Phys.}} \textbf{355} (2017), 741--766, \href{https://arxiv.org/abs/1607.07831}{arXiv:1607.07831}.

\bibitem{Rains11}
Rains E.M., An isomonodromy interpretation of the hypergeometric solution of
 the elliptic {P}ainlev\'e equation (and generalizations), \href{https://doi.org/10.3842/SIGMA.2011.088}{\textit{SIGMA}}
 \textbf{7} (2011), 088, 24~pages, \href{https://arxiv.org/abs/0807.0258}{arXiv:0807.0258}.

\bibitem{RGTT01}
Ramani A., Grammaticos B., Tamizhmani T., Tamizhmani K.M., Special function
 solutions of the discrete {P}ainlev\'e equations, \href{https://doi.org/10.1016/S0898-1221(01)00180-8}{\textit{Comput. Math.
 Appl.}} \textbf{42} (2001), 603--614.

\bibitem{Sakai01}
Sakai H., Rational surfaces associated with af\/f\/ine root systems and geometry of
 the {P}ainlev\'e equations, \href{https://doi.org/10.1007/s002200100446}{\textit{Comm. Math. Phys.}} \textbf{220} (2001),
 165--229.

\bibitem{Sakai05-2}
Sakai H., Hypergeometric solution of {$q$}-{S}chlesinger system of rank two,
 \href{https://doi.org/10.1007/s11005-005-0020-z}{\textit{Lett. Math. Phys.}} \textbf{73} (2005), 237--247.

\bibitem{Sakai05-1}
Sakai H., A {$q$}-analog of the {G}arnier system, \href{https://doi.org/10.1619/fesi.48.273}{\textit{Funkcial. Ekvac.}}
 \textbf{48} (2005), 273--297.

\bibitem{Sakai06}
Sakai H., Lax form of the {$q$}-{P}ainlev\'e equation associated with the
 {$A^{(1)}_2$} surface, \href{https://doi.org/10.1088/0305-4470/39/39/S13}{\textit{J.~Phys.~A: Math. Gen.}} \textbf{39} (2006),
 12203--12210.

\bibitem{Suzuki15}
Suzuki T., A {$q$}-analogue of the {D}rinfeld--{S}okolov hierarchy of
 type~{$A$} and {$q$}-{P}ainlev\'e system, in Algebraic and Analytic Aspects
 of Integrable Systems and {P}ainlev\'e Equations, \href{https://doi.org/10.1090/conm/651/13037}{\textit{Contemp. Math.}},
 Vol. 651, Amer. Math. Soc., Providence, RI, 2015, 25--38, \href{https://arxiv.org/abs/1105.4240}{arXiv:1105.4240}.

\bibitem{Suzuki17}
Suzuki T., A reformulation of the generalized {$q$}-{P}ainlev\'e~{VI} system
 with {$W(A^{(1)}_{2n+1})$} symmetry, \href{https://doi.org/10.1093/integr/xyw017}{\textit{J.~Integrable Syst.}} \textbf{2}
 (2017), xyw017, 18~pages, \href{https://arxiv.org/abs/1602.01573}{arXiv:1602.01573}.

\bibitem{Takenawa03}
Takenawa T., Weyl group symmetry of type {$D^{(1)}_5$} in the
 {$q$}-{P}ainlev\'e~{V} equation, \href{https://doi.org/10.1619/fesi.46.173}{\textit{Funkcial. Ekvac.}} \textbf{46}
 (2003), 173--186.

\bibitem{Yamada09-2}
Yamada Y., A {L}ax formalism for the elliptic dif\/ference {P}ainlev\'e equation,
 \href{https://doi.org/10.3842/SIGMA.2009.042}{\textit{SIGMA}} \textbf{5} (2009), 042, 15~pages, \href{https://arxiv.org/abs/0811.1796}{arXiv:0811.1796}.

\bibitem{Yamada09-1}
Yamada Y., Pad\'e method to {P}ainlev\'e equations, \href{https://doi.org/10.1619/fesi.52.83}{\textit{Funkcial. Ekvac.}}
 \textbf{52} (2009), 83--92.

\bibitem{Yamada11}
Yamada Y., Lax formalism for {$q$}-{P}ainlev\'e equations with af\/f\/ine {W}eyl
 group symmetry of type {$E^{(1)}_n$}, \href{https://doi.org/10.1093/imrn/rnq232}{\textit{Int. Math. Res. Not.}}
 \textbf{2011} (2011), 3823--3838, \href{https://arxiv.org/abs/1004.1687}{arXiv:1004.1687}.

\bibitem{Yamada14}
Yamada Y., A simple expression for discrete {P}ainlev\'e equations, in Novel
 Development of Nonlinear Discrete Integrable Systems, \textit{RIMS
 K\^oky\^uroku Bessatsu}, Vol. B47, Res. Inst. Math. Sci. (RIMS), Kyoto, 2014,
 87--95.

\bibitem{Yamada17}
Yamada Y., An elliptic {G}arnier system from interpolation, \href{https://doi.org/10.3842/SIGMA.2017.069}{\textit{SIGMA}}
 \textbf{13} (2017), 069, 8~pages, \href{https://arxiv.org/abs/1706.05155}{arXiv:1706.05155}.

\end{thebibliography}
\end{document}